\documentclass[11pt]{article}

\usepackage{amssymb,amsmath,amsfonts,amsthm,authblk}

\usepackage{srcltx}

\date {}

\newcommand{\N}{\mathbb{N}}
\newcommand{\Z}{\mathbb{Z}}

\title{On entropies of block-gluing subshifts}

\author[1]{Svetlana Puzynina \thanks{The first author  has been partially supported by Russian Foundation of Basic Research (grant 20-01-00488).}}

\author[2]{Mathieu Sablik}

\affil[1]{Saint Petersburg State University, Russia and Sobolev Institute of Mathematics, Russia, s.puzynina@gmail.com}
\affil[2]{Universit\'e Toulouse Paul Sabatier, France, Mathieu.Sablik@math.univ-toulouse.fr}

\newtheorem{proposition}{Proposition}
\newtheorem{corollary}{Corollary}
\newtheorem{lemma}{Lemma}
\newtheorem{theorem}{Theorem}

\theoremstyle{definition}
\newtheorem{definition}{Definition}
\newtheorem{example}{Example}
\newtheorem{remark}{Remark}
\newtheorem{conjecture}{Conjecture}
\newtheorem{question}{Question}

\begin{document}

\maketitle

\begin{abstract}
A subshift $X$ is called $c$-block gluing if for any integer $n
\geq c$ and any two blocks $u$ and $v$ from the language of $X$ there exists an element of $X$ which has occurrences of $u$ and $v$ at distance $n$. In this note we study the topological entropies of $c$-block gluing binary one-dimensional subshifts. We define the set $R_c$ to be the set of entropies of all $c$-block-gluing subshifts, and $R=\cup_{c\in \mathbb{N}} R_c$. We show that the set $R$ is dense, while $R_1$ and $R_2$ are not; in particular, they have isolated points. We conjecture that the same holds for any $c$. \end{abstract}

%\tableofcontents

\section{Introduction}

Topological entropy is the most important numerical invariant of a topological dynamical system~\cite{AKA}. For one dimensional subshift, topological entropy is given by the exponential growth rate of admissible words of the subshift of larger and larger size. For a class of dynamical systems, it is interesting to determine which numbers can be obtained as entropies. For example, the set of subshifts of finite type is countable, so not all non-negative real numbers can be realized as entropies of these systems. In the case of $\Z$-subshifts of finite type the set of their entropies is the set of logarithm of Perron numbers~\cite{LM95}, and in the case of  $\Z^2$-subshifts of finite type it is the set of upper-computable numbers~\cite{Hochman-Meyerovitch-2010}.

A general point of view can be to specify the algorithmic complexity of the set of numbers realized as entropies depending on the dynamical properties of the subshift of finite type (see~\cite{Herrera-Gangloff-Rojas-Sablik-2020} for more details). An interesting dynamical property for subshift of finite type is the notion of $c$-block-gluing. This is a mixing-type property introduced in~\cite{Boyle-Pavlov-Schraudner-2010} which means that there exists a constant $c$ such that for any two blocks in the language of the subshift, the pattern
obtained by gluing these two blocks at any distance greater or equal to $c$  between them is also in the language of the subshift. In fact, the entropies of $\Z^2$-subshifts of finite type which are block gluing are computable~\cite{PS15}; however, it is still not known if all computable numbers can be realised. In fact it is possible to extend the notion of block gluing by   adding a gap function giving the distance which allows to concatenate two patterns. This has been studied for $\Z^2$-subshifts of finite type in~\cite{GS20}, and in~\cite{GH19} the authors give a precise threshold in the case of one dimensional subshifts with decidable languages. Above this threshold all upper-computable numbers can be realized (it is the same that the class of subshifts with decidable languages without dynamical constraints). Under this threshold, and in particular for constant block gluing, only computable numbers can be obtained, but it is not known if all computable numbers can be realized. In fact, in the case of constant block gluing, the combinatorial constraints which appear in the study of  realizable entropies are not well understood.

In this paper we are interested in  understanding more precisely the possible entropies of block gluing $\Z$-subshift. We define a set $R_c$ to be the set of entropies of all $c$-block-gluing subshifts, and $R=\cup_{c\in \mathbb{N}} R_c$. We are interested in characterizing these sets. We show that the set $R$ is dense, while $R_1$ and $R_2$ are not; in particular, $1$- and $2$-block gluing subshifts with minimal entropies give isolated points in $R_1$ and $R_2$, respectively. We conjecture that the same holds for any $c$. However, we also show that the sets $R_c$ have many accumulation points.

\section{Basic notions}

Let $A$ be a finite set called an  \emph{alphabet}. A
\emph{configuration} is a sequence of elements of $A$ indexed by
$\mathbb{Z}$, we let $A^\mathbb{Z}$ denote the set of all
configurations. It is a compact space for the product topology. A \emph{subshift} is a closed subset of $A^\mathbb{Z}$ invariant for the shift map $\sigma:A^\mathbb{Z}\to A^\mathbb{Z}$ defined by $\sigma(x)_i=x_{i+1}$ for all $x\in A^\mathbb{Z}$ and $i\in\Z$.

The \emph{language} of a subshift $X\subset A^\mathbb{Z}$, denoted by $L_X$, is the set of words defined as follows: a finite word $v$ is in $L_X$ if there exists a configuration in $X$ where the word $v$ occurs. The set of words of size $n\in\mathbb{N}$ which occur in $X$ is denoted by $L_X(n)$.

Equivalently, a subshift can be defined by the set of forbidden blocks which cannot appear in a configuration of $X$. If this set is finite, $X$ is called a \emph{subshift of
finite type} (SFT for short). The \emph{order} of an SFT is the smallest integer $r$ such that the subshift can be defined by
forbidden words of length $r$. If an SFT $X$ is defined by a finite set $S$ of forbidden words, we denote it by $X=SFT(S)$.

 The subshift $X$ is \emph{$c$-block
gluing} if for all $u, v\in L_X$, for all $n\geq c$ there exists $w$ such that $|w|=n$ and $uwv\in L_X$.

\begin{example} Let $d=1$ and let $X= SFT(11)$. Clearly, $X$ is
1-block-gluing, since we can glue at any distance with all-$0$ word.\end{example}

The following lemma states that to prove that a subshift is $c$-block-gluing, it is enough to check that we can paste any two words from its language at distance $c$:

\begin{lemma} \label{lem:c} A subshift $X$ is $c$-block
gluing if and only if for all $u, v\in L_X$,  there exists $w$ such that $|w|=c$ and $uwv\in L_X$. \end{lemma}

\begin{proof} The ``only if'' direction follows from the definition of $c$-block-gluing. To prove the ``if'' direction, it is enough to notice that we can glue any two words $u$ and $v$ from the language at distance $n>c$. Since $u$ is in the language, it can be extended to the right by a word of length $n-c$, i.e, there exists $w$, $|w|=n-c$, such that $uw$ is in the language of the subshift. Since we can glue $uw$ and $v$ at distance $c$, this shows that we can glue $u$ and $v$ at distance $n$. \end{proof}

If a subshift is not of finite type, it can still be defined by a (infinite) set $S$ of prohibited words. So, given $S\subseteq \Sigma^*$, we let $X_S$ denote a subshift containing all configurations avoiding $S$.

The \emph{entropy} of a subshift $X$ is defined as:
$$h(X) = \inf \frac {\log(|L_X(n)|)}{n}.$$

We say that two subshifts $X$ and $Y$ are \emph{conjugate} if there exists a continuous function from $X$ to $Y$ which is bijective and which commutes with $\sigma$. It is well known that two conjugate subshifts have the same entropy (see~\cite{LM95} for
details).

\section{General properties of $c$-block-gluing subshifts}

\subsection{Existence of periodic points}

The following proposition states that in the one-dimensional case, each $c$-block gluing subshift contains periodic points.
\begin{proposition}
Each $c$-block gluing subshift $X\subseteq A^{\mathbb{Z}}$
contains a periodic point.
\end{proposition}

\begin{proof}
% In fact, for the proof it is enough to glue factors at distance $c$.
We start with a letter $a$, and see how we can glue it with itself
at distance $c$: $a*^ca$. We have several possibilities, at most
$|A|^c$. If we have only one, we have a periodic point. If we have
several, choose any of them (a word $w_1$ of length $c$), and
continue the process with the word $aw_1a$ instead of $a$. Now the
choice is a subset of the previous set of possibilities. Again, if
$w_1$ is still possible, continue the process with it, or if not, choose
any other possible word $w_2$, and continue the process.
Either the set of choices is stabilised and sticking to the same
word we get a periodic point, or at some point there will be no
choice, and so we also get a periodic point.
\end{proof}

\subsection{Description of $c$-block gluing with Rauzy graph}

Let $G=(V,E,\lambda)$ be a graph with the set $V$ of vertices  and the set $E$ of edges labeled by
$\lambda:E\to A$ for all $e\in E$. We let $\mathbf{i}(e)$ denote the initial
vertex and $\mathbf{t}(e)$ the terminal vertex. We can  define the graph subshift
\begin{eqnarray*}
Y_G=\{x\in A^Z: \textrm{ there exists }(e_i)_{i\in\Z}\in E^\Z \textrm{ such that } \\ \mathbf{t}(e_i)=\mathbf{i}(e_{i+1})  \textrm{ and }  \lambda(e_i)=x_i \textrm{ for all } i\in\Z \}.\end{eqnarray*}
Let $X$ be a subshift, define $G_n(X)=(V_n(X),E_n(X))$ the Rauzy
graph of order $n$ such that the vertices are $V_n(X)=L_X(n-1)$
and for $u=u_1\dots u_{n-1}$ and $v=v_1\dots v_{n-1}$ in
$L_X(n-1)$ one has $(u,v)\in E_n(X)$ if and only if $u_1\dots
u_{n-1}v_{n-1}=u_1v_1\dots v_{n-1}\in L_X(n)$ and the edge $(u,v)$
is labeled by the letter $v_{n-1}\in A$. See~\cite{Fogg} for more details.

Let $X$ and $Y$ be two subshifts,  the \emph{Hausdorff
distance} is defined as follows: $$d_H(X,Y)=2^{-\min\{n:L_X(n)\ne L_Y(n)\}}.$$ The set of
subshifts endowed with distance $d_H$ is compact. Denote
$X_n\overset{H}{\longrightarrow}Y$ if
$d(X_n,Y)\underset{n\to\infty}{\longrightarrow}0$.

Clearly, $Y_{G_n(X)}$ is a subshift of finite type of order $n$ such that $L_X(n)=L_{Y_{G_n(X)}}(n)$; in other words, $d_H(X,Y_{G_n(X)})\leq 2^{-n}$,
$Y_{G_{n+1}(X)}\subset Y_{G_n(X)}$ and
$$X=\bigcap_{n\in\N}Y_{G_n(X)}.$$

 The properties of Rauzy graphs of $c$-block-gluing subshifts are summarized in the following propositions.

\begin{proposition}\label{prop:SFTcbg}
Let $X$ be an SFT of order $n$. Then $X$ is $c$-block gluing if and only if for every pair $(u,v)$ of vertices of $G_n(X)$ % and for
%every $k$ with $n+c-1\leq k \leq D$, where $D$ is the diameter of $G_n(X)$,
there is a path of length $n+c-1$ from $u$ to $v$.
\end{proposition}

\begin{proof} The necessity of the condition is obvious by definition of $c$-block-gluing subshift: we must be able to glue any two words of length $n-1$ (they correspond to vertices of $G_n(X)$) at distance $c$ (this corresponds to a path of length $n+c-1$).

 To prove the sufficiency of the condition, by Lemma \ref{lem:c} it is enough to prove that for any two words $u$, $v$ %of length $m$
from the language of $X$ we can glue them at distance $c$, i.e., there exists a word $w$ of length $c$ such that $uwv\in L_X$.
If $u$ (resp., $v$) is of length $n-1$, it corresponds to a vertex $u$ (resp., $v$) of $G_n(X)$.  If $u$ (resp., $v$) is shorter than $n-1$, take any vertex $u'$ (resp., $v'$) whose label has $u$ as its suffix (resp., $v$ as a prefix). If $u$ (resp., $v$) is longer than $n-1$, consider a path in $G_n(X)$ corresponding to $u$, i.e. a path of length $|u|-n+1$ with labels from the prefix of $u$ and its final vertex marked by a suffix of $u$ of length $n-1$ (resp., a path corresponding to $v$, i.e. a path of length $|v|-n+1$ with labels from the suffix of $v$ and its initial vertex marked by a prefix of $v$ of length $n-1$); we let $u'$ (resp., $v'$) denote the label of its final (resp., initial) vertex. By the condition of the lemma, there is a path of length $n+c-1$ connecting $u'$ and $v'$. Let $w$ be the word marking the path. It has $v'$ as its suffix, i.e., $w=w'v'$ with $|w'|=c$. In each of the cases above, we have the word $uw'v$ in the language of $X$.
\end{proof}

\begin{proposition}
A subshift $X$ is $c$-block gluing if and only if for all
$n\in\N$, for every pair $(u,v)$ of vertices of $G_n(X)$ and for every $k\geq c+n-1$, there is a path from $u$ to $v$ of size $k$.
\end{proposition}

\begin{proof}
If $X$ is $c$-block gluing then for all $u,v\in V_X(n)=L_X(n-1)$ there is a configuration $x\in X\subset Y_{G_n(X)}$ such that $x_{[0,|u|-1]}=u$ and $x_{[|u|+k,|u|+k+|v|-1]}=v$ for $k\geq c$ so there is a path from $u$ to $v$ of size $k+n-1$.

Reciprocally, given two words $u,v\in L(X)\subset L(Y_{G_n(X)})$ and $k\geq c$ then for every $n\in\N$ there exists $x^n\in Y_{G_n(X)}$ such that $x^n_{[0,|u|-1]}=u$ and $x^n_{[|u|+k,|u|+k+|v|-1]}=v$ since there is a path from $x^n_{[|u|-n-2,|u|-1]}=u$ and $x^n_{[|u|+k,|u|+k+n-1]}=v$ of size $k+n-1$. By compactness there exists $x\in\cap_nY_{G_n(X)}=X$  such that $x_{[0,|u|-1]}=u$ and $x_{[|u|+k,|u|+k+|v|-1]}=v$. Thus $X$ is $c$-block gluing.
\end{proof}

The previous two propositions imply the following:

\begin{proposition}\label{prop:Gn}
 If $X$ is a $c$-block gluing subshift then $Y_{G_n(X)}$ is a $c$-block gluing  subshift of finite type such that $d_H(X,Y_{G_n(X)})\leq 2^{-n}$ for all $n\in\N$.
\end{proposition}

%Indeed, consider any two words $u$, $v$ of length $m$ from the language of $X$ and an integer $l\geq c$.

%If $m=n-1$, then these two factors correspond to two vertices of $G_n(X)$. If $l\leq D+n-1$, the path of length $l+n-1$ between them is marked by a word $w$ of length $l+n-1$ with $v$ as suffix, i.e., $w=w'v$ with $|w'|=l$. This gives the word $uw'v$ in the language of $X$. If $l>N$, consider any path starting from $u$ of length $l-D$ with a corresponding word $w'$. Let $u'$ be the vertex attained from $u$ using this path. Now by the conditions of the proposition there exists a path of length $D$ connecting $u'$ with $v$, let $w''$ the word corresponding to this path. So, $uw'w''v\in L_X$, and $w'w''$ is of length $l$.

%If $m<n$, take any words $u'$ and $v'$ such that $u$ is a suffix of $u'$ and $v$ is a prefix of $v'$. By what we proved above, there is a path of length $l$ connecting $u'$ and $v'$, and hence $u$ and $v$.

%If $m>n$, consider a path in $G_n(X)$ corresponding to $u$ (we let $u'$ denote its final vertex) and a path corresponding to $v$ (with $v'$ its first vertex). By what we proved above, there is a path of any length greater or equal to $c$ connecting $u'$ and $v'$. Denoting the word corresponding to this path, we have $uwv\in L(X)$. \end{proof}

\begin{remark}In particular, due to Proposition \ref{prop:SFTcbg}, one can algorithmically decide if a given SFT is $c$-block-gluing.

%Note that in general it is undecidable to find out if a subshift with a decidable language is $c$-block gluing. Indeed, given a Turing machine $\mathcal{M}$, one considers the subshift $X_\mathcal{M}\subset\{0,1\}^Z$ where the forbidden words are $10^n10^m$ if $\mathcal{M}$ halts in $n$ steps on the empty input and $m<n$. Clearly, the language of $X_\mathcal{M}$ is decidable. If $\mathcal{M}$ does not halt, then $X_\mathcal{M}$ is $0$-block gluing ($10^n1$ can be followed by any words. If $\mathcal{M}$ halts in $t$ steps, then $10^t1$ and $1$ cannot be pasted at a distance less than $t$. Thus $X_\mathcal{M}$ is $0$-block gluing iff $\mathcal{M}$ does not halt.

Note that in general it is undecidable to find out if a subshift with a decidable language is $c$-block gluing. Indeed, given a Turing machine $\mathcal{M}$, one considers the subshift $X_\mathcal{M}\subset\{0,1\}^Z$ where the forbidden words are $10^n10^m$ if $\mathcal{M}$ halts in $n$ steps on the empty input and $m<n$. Clearly, the language of $X_\mathcal{M}$ is decidable. If $\mathcal{M}$ does not halt, then $10^n1$ can be followed by any words. If $\mathcal{M}$ halts in $t$ steps, then $10^t1$ and $1$ cannot be pasted at a
distance less than $t$. Thus $X_\mathcal{M}$ is $0$-block gluing iff $\mathcal{M}$ does not halt.

However, we remark that the property is co-semi-decidable since to find out that a subshift is undecidable it suffices to find two words $u$, $v$ and a distance $k>c$ such that none of the words $u A^k v$ is in the language of the subshift.
\end{remark}

\section{General properties of the spectrum} \label{section:general}

Let $c\in\N$, define \emph{the spectrum of order $c$}:
$$R_c=\left\{h(X):X\textrm{ $c$-block gluing subshift}\right\}.$$

Clearly, $R_c\subset R_{c+1}$; define  \emph{the spectrum}
$R=\bigcup_{c\in\N} R_c$ the set of possible entropies of constant
block gluing subshifts.

The main problematic is to understand the structure of $R_c$ and
$R$.

%{\bf{Question}}: What can we say about $h(X)$ in $R_c$?

\begin{proposition}\label{prop:lowerbound} Let $X$ be a $c$-block-gluing subshift with positive entropy, then $h(X)\geq \frac{\log (|L_X(k)|)}{c+k}$ for each $k$. \end{proposition}

In particular, taking $k=c$, we get $h(X)\geq \frac{\log
(c+1)}{2c}$. Indeed, $|L_X(k)|\geq c+1$ since $X$ is not reduced to
periodic points.

\begin{proof}
 For each length $k$, we can put any two words from
$L_X(k)$ at distance $c$, hence the inequality $|L_X(2k+c)|\geq
|L_X(k)|^2$. Similarly, we get $|L_X(l(k+c)|\geq |L_X(k)|^l$ for
each integer $l$. So

$$h(n)=\lim_{l\to\infty}\frac{\log |L_X(l(k+c))|}{l(k+c)} \geq  \frac{\log|L_X(k)|}{k+c}.$$
\end{proof}

 \begin{corollary}
 One has $R_c\subset\{0\}\cup\left[\frac{\log(2)}{c+1},+\infty\right[$.
  \end{corollary}
\begin{proof}
 Let $X\subset A^{\Z^d}$ be a $c$-block-gluing subshift such that $\log|L_X(1)|\geq2$, then $h(X)\geq \frac{\log(2)}{c+1}$.

\end{proof}

\begin{proposition}
 $R$ is dense in $[0,+\infty)$.
\end{proposition}
\begin{proof}
In Chapter 11 of \cite{LM95} it is shown that the possible entropies of SFT's are logarithms of Perron Numbers. In the proof, given a Perron Number, they construct a SFT which realizes the associated entropy; moreover, this SFT can be topoplogically mixing, so it is $c$-block-gluing for some $c$. As Perron numbers are dense in $[1,+\infty[$, we deduce the density of $R$.
%Using density of entropy of transitive subshift of finite type.
\end{proof}

\begin{proposition}
 Let $(h_n)_{n\in\N}$ be a sequence of elements of $R_c$, then $\inf_nh_n\in R_c$.
\end{proposition}
\begin{proof}
 For all $n$, consider $X_n$ a $c$-block gluing subshift such that $h_n=h(X_n)$. By compactness, we can extract a converging subsequence $X_{n_k}$, so, slightly abusing notation, we assume $X_n\overset{H}{\longrightarrow}X$.

 Let $u,v\in L_X$. There exists $N$ such that $u,v\in L_{X_n}$ for all $n\geq N$. Since $X_n$ is $c$-block-gluing, for $c'\geq c$ there exists $w_n$ such that $|w_n|=c'$ and $uw_nv\in L_{X_n}$. Let $w$ be a word which appears infinitely often in the sequence $(w_n)_{n\in\N}$, we deduce that $uwv\in L_X$. Thus $X$ is also $c$-block-gluing.

Moreover, since for all $k\in\N$ there exists $N\in\N$ such that
$L_{X_n}(k)=L_X(k)$ for all $n\geq N$, one has
$$h(X)=\inf_{k\in\N}\frac{\log(|L_X(k)|)}{k}=\inf_{k\in\N}\inf_{n\in\N}\frac{\log(|L_{X_n}(k)|)}{k}=\inf_{n\in\N}\inf_{k\in\N}\frac{\log(|L_{X_n}(k)|)}{k}=\inf_{n\in\N}h(X_n).$$
Thus $h(X)=\inf_nh_n\in R_c$.
 \end{proof}

\begin{proposition}
Let $h\in R_c$ and $\epsilon>0$ such that $]h, h+\epsilon[\cap
R_c=\emptyset$. Then %there exists a $c$-block gluing subshift of
%finite type such that $h=h(X)$.
each $c$-block-gluing subshift  with entropy $h$ is a subshift of finite type.
\end{proposition}
\begin{proof}
 Let $X$ be a $c$-block-gluing subshift such that $h=h(X)$. There exists $n\in\N$ such that
 $$h(X)\geq\frac{\log(|L_X(n)|)}{n}-\frac{\epsilon}{2}=\frac{\log(|L_{Y_{G_n(X)}}|)}{n}-\frac{\epsilon}{2}\geq h(Y_{G_n(X)})-\frac{\epsilon}{2}.$$

As $Y_{G_n(X)}$ is $c$-block gluing by Proposition \ref{prop:Gn}  and $]h,h+\epsilon[\cap
R_c=\emptyset$, one deduces that $h(Y_{G_n(X)})=h$ and in fact
$Y_{G_n(X)}=X$.

\end{proof}

\section{Maximal subshift}

% \begin{definition} A $c$-block-gluing subshift $X=SFT(S)$, where $S\subseteq A^*$ is a finite set of prohibited words, is called {\emph{small}} (or something) if for any $S'\subseteq A^*$, $S\cap S'=\emptyset$, we have that $Y=SFT (S\cup S')$ is not $c$-block gluing. (Technically, if $S'$ is infinite, then $Y$ is not an SFT.)
% \end{definition}

We will make use of the following notion:

\begin{definition} A $c$-block-gluing subshift $X$ is \emph{maximal} of order $n$ if for all $Y\subset X$ a $c$-block-gluing subshift such that $L_X(n)=L_Y(n)$, one has $X=Y$.
 \end{definition}

 We call it maximal in the sense that we prohibited maximal set of factors to keep the property of $c$-block-gluing.

 Clearly, if a subshift is maximal of order $n$, it is maximal of order $n'$ for each $n'>n$. The following examples show that the converse is not true.

%{\color{red}S:Examples added}
\begin{example}\label{ex:SFTmax}$X=SFT(11,101)$ is maximal $2$-block gluing subshift of order $3$. To prove that, we must show that for any $2$-block gluing subshift such that $L_X(3)=L_Y(3)$ any word $v$ from $L_Y$  can be extended by both $0$ and $1$, unless $va$ has a forbidden factor as a suffix  (i.e, $va\notin L_X$). If $v=v'1$, then $v1$ is forbidden and $v$ must be continued by $0$, so $v0\in L_Y$. If $v=v'0$, then either $v=v''00$ or $v=v''10$. In the first case considering $v''**1$ and $v''0**1$, we get that both $v0$ and $v1$ belong to $L_Y$, and in the second case $v1$ is forbidden, so $v0\in L_Y$.

%To show that it is not maximal of order $2$, %consider a $2$-block gluing subshift %$Z=SFT(11,10101)$. Clearly, we have $L_2(X)=L_2(Z)$, %but $X\neq Z$, so $X$ is not maximal of order $2$.

%Actually, in the same way one can show that $Z$ is %maximal of order 5 according to our definition.
\end{example}

\begin{example}$Z=SFT(11,10101)$ is maximal $2$-block gluing of order $5$, but is not of order $2$. The proof of maximality order $5$ is similar to the previous example. Clearly, it is not maximal of order $2$ since for $2$-block gluing subshift $X=SFT(11,101)$ we have $L_X(2)=L_Z(2)$, $X\subset Z$, but $X\neq Z$.\end{example}

 \begin{proposition} \label{prop:maximal}
$SFT(1^k)$ is a maximal $1$-block-gluing subshift of order $k$.

$SFT(10^k1)$ is a maximal $1$-block-gluing subshift of order
$k+2$.  \end{proposition} %{\color{red}S: check the order}

\begin{proof}
 Consider first $SFT(1^k)$. Let $Y\subset X$ be $1$-block-gluing such that $L_Y(k)=L_X(k)$. We are going to prove by induction that $L_Y(n)=L_X(n)$ for all $n\in\N$. Let $x\in L_Y(n)$ and $L_X(n)$,  we are going to prove that $x0$ and $x1$ are also $L_Y(n+1)$ if they are in $L_X(n+1)$:
 \begin{itemize}
\item By definition of $1$-block-gluing $x*1^{k-1}\in L_Y$ for some $*\in\{0,1\}$. Since $1^k\notin L_Y$, $*$ cannot be $1$ so $x0\in L_Y$.
\item To show that $x1\in L_Y$ if $x1\in L_X$ we need to consider two cases:
\begin{itemize}
\item If $x=x'0$, by $1$-block-gluing $x'*1^{k-1}\in L_Y$ for some $*\in\{0,1\}$ and $*$ cannot be $1$. Thus $x1\in L_Y$.
\item If $x=x'01^i$ with $i\geq 1$, then, since $1^k\notin L_Y$ and $x1\in L_X$, we have $i\leq k-2$. By  $1$-block-gluing $x'*1^{k-1}\in L_Y$ for some $*\in\{0,1\}$ and $*$ cannot be $1$. Thus $x1=x01^{i+1}\in L_Y$.
\end{itemize}

\end{itemize}

 Now consider $SFT(10^k1)$. Let $Y\subset X$ be $1$-block-gluing such that $L_Y(k+2)=L_X(k+2)$. We are going to prove by induction that $L_Y(n)=L_X(n)$ for all $n\in\N$. Let $x\in L_Y(n)$ and $L_X(n)$,  we are going to prove that $x0$ and $x1$ are also $L_Y(n+1)$ if they are in $L_X(n+1)$:
 \begin{itemize}
\item By definition of $1$-block-gluing $x*0^{k}1\in L_Y$ for some $*\in\{0,1\}$ and $*$ cannot be $1$, so $x0\in L_Y$.
\item To show that $x1\in L_Y$ if $x1\in L_X$, we need to consider different cases:
\begin{itemize}
\item If $x=x'10^i$ with $0\leq i\leq k-1$, by  $1$-block-gluing $x'10^i*0^{k-i-1}1\in L_Y$ for some $*\in\{0,1\}$ and $*$ cannot be $0$. Thus $x1\in L_Y$.
\item If $x=x'10^k$ then $x1\notin L_X$.
\item If $x=x'0^i$ with $ i\geq k+1$, by  $1$-block-gluing $x'0^{i-k-1}*0^{k}1\in L_Y$ for some $*\in\{0,1\}$ and $*$ cannot be $1$. Thus $x1=x'0^{i-k-1}00^{k}1\in L_Y$.
\end{itemize}\end{itemize}\end{proof}

\begin{remark} \label{remark:maximal} We remark that in fact to prove that an SFT $X$ of order at most $n$ is maximal of order $n$ one can show that for any $c$-block-gluing subshift $Y\subset X$ with $L_X(n)=L_Y(n)$ for each $v\in L(X)$ of length at least $n$ and a letter $a$ such that $va\in L_X$, one has to have $va\in L_Y$. A sufficient condition is given by the following. Since $X$ is $c$-block gluing,  for each $v,u\in L_X(n)$, each $a\in A$ such that $va \in L_X(n+1)$ and each prefix $v'$ of $v$ one can glue $v'*^cu$. If in addition for each $v\in L_X(n)$ and each $a\in A$ with $va\in L_X$ one has to have $va$ as prefix of  $v'*^cu$ for some $v'\in\mbox{pref}(v)$ and some $u\in L_X(n)$ (clearly, this can be verified with a final case study), then $X$ is maximal. Indeed, for $Y\subset X$ to be $c$-block gluing it is necessary to have $va\in L_Y$: for $|v|=n$ it is given by the condition, and for a longer word $v$ it is enough to to apply the condition for its suffix of length $n$.
%One can usually do this by choosing carefully a word $u\in L(X)$ and a distance $n\geq c$ to glue (a prefix of) $v$ with $u$, as it is done in Proposition \ref{prop:maximal} and Example \ref{ex:SFTmax}.
\end{remark}

In the following proposition we prove the maximality of the subshifts which are candidates to have minimal entropy. We suggest that both 0 and 1 are in the language of the subshift to exclude trivial subshifts containing only 0's or only 1's.
\begin{proposition} If $Y$ is a $c$-block-gluing subshift containing letters $0$ and $1$, such that $11, 101, 1001, \dots, 10^{c-1}1\notin L_Y$, then
$$Y= SFT\{11, 101, 1001, \dots, 10^{c-1}1\}.$$
In other words, $Y$ is maximal of any order. \end{proposition}

\begin{proof} We will prove that each $u$ in the language of $SFT\{11, 101, 1001, \dots, 10^c1\}$ must also be in $L_Y$. We let $k$ denote the length of $u$.

For length $k=2, \dots, c+2$ the proof follows from the complexity
argument: For $k=2, \dots, c+1$ we only have $k+1$ factors (which
is minimal possible, otherwise we have only periodic words by Morse and Hedlund theorem, which is impossible for block-gluing subshifts). For $k=c+2$ we have
$k+2$ factors, which is also minimal possible. Indeed, for the length $c+2$ we must have a factor $0^{c+2}$ (gluing $00*^{c}1$), all $c+2$ factors with exactly one occurrence of $1$ (consider $0*^{c}10^c$ or symmetric factors), and  $10^c1$ (gluing $1*^c1$).

For length $k\geq c+2$ we prove by induction by the length $k$ of $u$ that each $u\in SFT\{11, 101, 1001, \dots, 10^c1\}$  must also be in $L_Y$.
Consider the tail of $0$'s %last occurrence of $1$
in $u$: $u=u'0^i$, so that $u'$ either ends with $1$ or is empty. If $i=0$,
then $u=u''0^c1$ and we get $u$ by applying the definition of
$c$-block-gluing to $u''*^c1$. If $i\geq c$, then we get $u$
considering the factor $u'0^{i-c}*^c1$. If $i< c$, then we get
$u$ considering the factor $u'*^c1$. The proposition is proved.
\end{proof}

\begin{example} Other examples of maximal subshifts:
SFT$(000,010)$ for $c=1$ of order $3$, SFT$(11, 10001, 100001)$ for $c=2$ of order $6$. The proofs are similar to the propositions above.
\end{example}

\section{Particular cases}\label{section:particular}

In this section, we provide some results about minimal entropies, isolation points and accumulation points for for small values of $c$.
In all this section $A=\{0,1\}$.

\subsection{Minimal entropy for $c=1$}

\begin{theorem}\label{thm:min_ent}Minimal positive entropy of 1-block-gluing subshift  corresponds to $SFT\{11\}$, and there is no 1-block-gluing subshift with entropy between $h(SFT\{11\})=\log\frac{1+\sqrt5}{2}$ and $\frac{\log(7)}{4}$. In other words,

$$R_1\subset\left\{0,\log\frac{1+\sqrt5}{2}\right\}\cup\left[\frac{\log(7)}{4},+\infty\right[$$

\end{theorem}

We need several lemmas for the proof.

First we remark that due to Proposition \ref{prop:maximal}, if $Y$
is a 1-block-gluing subshift and $11\notin L_Y$, then $Y=
SFT\{11\}$.

%\begin{lemma} If $Y$ is a 1-block-gluing subshift and $11\notin L_Y$, then $Y= SFT\{11\}$. \end{lemma}

%\begin{proof}
%As $11\notin L_Y$, clearly $Y\subset SFT\{11\}$.

%Reciprocally, by induction on the length of the factors we show that all factor of $SFT\{11\}$ is in $L_Y$. Clearly, $00$, $10$, $01 \in \in L_y$. Suppose that this holds for all factors of length less than $n$. Consider a factor $x\in SFT\{11\}$ of length $n-1$. It is enough to prove that its extensions by one letter to the right, if in $ SFT\{11\}$, also belong $L_y$. If $x$ ends with a $1$, by definition of $1$-block gluing $x*1\in L_Y$ for some $*\in\{0,1\}$. $*$ cannot be 1 since $11\notin S_Y$, so $*=0$ and $x0$, the only possible extension of $x$ to the right, belongs to $S_Y$. If $x$ ends with $0$, then consider $x*1$ and $x_0 \cdots x_{n-2}*1$. In the first case $*=0$ and so $x0\in L_Y$, in the second case also $*=0$, so $x1\in L_Y$.
%\end{proof}

\begin{lemma}\label{lem:MinLY} If $00, 11\in L_Y$, then $|L_Y(3)|\geq 7$. \end{lemma}

\begin{proof}
Consider a factor $00*11$, which must be in $Y$ for $*=0$ or $*=1$
by the definition of 1-block-gluing. So, $001, 011 \in Y$, and
either $000\in Y$ or $111\in Y$. Symmetrically, $110, 100 \in Y$.

Now if $000\notin Y$, then $111\in Y$, and consider $10*01\in Y$.
In this case $*=1$, and we get $010, 101\in Y$, which gives 7
factors. Symmetrically if $111\notin Y$.

If $010\notin Y$, then by what we proved above $000, 111\in Y$. Considering the word $01*01$, we ge that $*=1$ and hence $101$ is a factor. This also gives 7 factors. Symmetrically for the case $101\notin
Y$.
\end{proof}

\begin{lemma} $h(SFT\{11\})=\log\frac{1+\sqrt5}{2}$. \end{lemma}

\begin{proof} The entropy of an SFT is known to be equal to log of the largest eigenvalue of its Rauzy graph of the order of the subshift \cite{LM95}.
\end{proof}

Now we can prove Theorem \ref{thm:min_ent}.

\begin{proof}
If $11$ or $00$ is not in the langage of $Y$ then using
Lemma~\ref{lem:MinLY} with the lower bound from
Proposition~\ref{prop:lowerbound} with $k=3$ and $c=1$, we obtain
 $$h(Y)\geq\frac{\log(|L_Y(3)|)}{4}\geq\frac{\log(7)}{4}>h(SFT\{11\}).$$
 \end{proof}

\subsection{Accumulation points}

In this subsection we provide a series of examples of $c$-block gluing subshifts for different values of $c$; this includes several constructions for infinitely many accumulaton points for their entropies.
%\textbf{Question.} Did we prove that $X=SFT(11,101,1001)$ has minimal entropy? I think not --- the difficulty comes from the fact that it is not unique.

\begin{example} For $c\geq 2$, $i,j>2$ we have that SFT$(0^i,1^j)$
is $c$-block-gluing (the proof is direct: we show we can connect
any two factors without producing forbidden factors). As a
corollary, we get infinitely many points of accumulation of
entropy for a given $c$: each of SFT$(0^i)$ is a point of
accumulation, and the corresponding converging sequence is
SFT$(0^i,1^j)$ for $j\to\infty$.
\end{example}

% The next two examples do not work.

%{\bf{Example}.} Similar example for $c=1$ is given by $k>3$ we have that SFT$(1001,10^{2k}1)$: it is 1-block gluing, and SFT$(1001)$ is  a point of accumulation.

%\bigskip

%{\bf{Example}.} The previous example for $c=1$ can be generalized to SFT $(10^{2l}1,10^{2k}1)$, where $k>2l$: it is 1-block gluing, and SFT$(10^{2l}1)$ is  a point of accumulation. More generally, $SFT(10^{2l_1}1,10^{2l_2}1, \dots )$, where $l_{i+1}>l_i$, is c-block-gluing and is a point of accumulation. The set of indices may as well be infinite, we still have a c-block gluing code and a point of accumulation -- in this case this is not an SFT though. TODO: check and give an accurate proof.

\begin{example} Similarly, for $c=3$  and any increasing sequence of integers $0<l_1<l_2< \ldots $ and a set
$S=\{10^{2l_1}1,10^{2l_2}1, \dots \}$, the subshift $X_S$ of words avoiding words from $S$ as factors  is $c$-block-gluing (the proof is straightforward). The set of indices may be finite (in this case we have an SFT) or infinite.  Clearly, this subshift gives a
point of accumulation both for a finite and infinite set $S$. % S: checked, correct. The proof is straightforward, no point of giving it
\end{example}

%Are there accumulation points for $c=1$ (other than full shift)?

%{\bf{Example}.} For $c=1$, take $S=\{101(10^{k_i})101|k_i\in K\subseteq\mathbb{N}\}$. Then $SFT(S)$ is 1-block gluing. This construction also gives a bunch of accumulation points. Proof to be added. S: Oops, they are all the same.

\begin{example} For $c=1$, take $S=\{10101, (101)^{k}\}$ for $k\in
\mathbb{N}$. Then $SFT(S)$ is 1-block-gluing: clearly, we can glue $u$ and $v$ with $0$'s unless $u$ ends with $1$ and $v$ starts with $1$, and in the latter case we can glue with $1$'s. This construction
gives an accumulation point SFT$(10101)$.
\end{example}

%{\bf{Question}}: Are there 1-block gluing subshifts which are not SFT?

\begin{example} Here we give an example of 1-block gluing subshift
which is not SFT (in fact, a family of examples and in fact not
even sofic). For $c=1$, take $S=101(01| 101)^*$, and take any
subset $S'\subseteq S$. It is not hard to see that $X_{S'}$ is
1-block gluing. Indeed, if we need to glue together $u$ and $v$
such that $u$ ends with 0 and/or $v$ begins with 0, then we can
glue with $0^+$ (we do not produce prohibited words since 00 is
not a factor of $S$ and no factor in $S$ starts or begins with
$0$). If we need to glue together $u$ and $v$ such that $u$ ends
with 1 and $v$ begins with 1, then we can glue with $1^+$ (we do
not produce prohibited words since $111$ is not a factor of $S$
and no factor in $S$ starts or begins with  $11$).

Now notice that we can choose $S'$ so that $X_{S'}$ is not SFT and
even not sofic. For example, take $S'=\{10101(101)^{k_i}01|k_i\in
K\}$ for $K$ the set of prime numbers. Suppose $L_{X_{S'}}$ is
regular. Then $L_{X_{S'}} \cap 10101(101)^{*}01=S'$ is regular.
But $S'$ is not regular: due to pumping lemma, the lengths of
words of a regular language must contain an infinite arithmetic
progression, but they do not for $S'$. A contradiction.
\end{example}

\subsection{Second entropy for $c=1$.}

The next theorem gives the second entropy for $c=1$ and shows that it is also an isolated point.

\begin{theorem}\label{thm:second}
For $c=1$, the second positive entropy is given by $SFT(101)$ and there is
no 1-block-gluing subshift with entropy between $h(SFT\{101\})$
and $\frac{\log(31)}{6}$.
\end{theorem}

With Theorem \ref{thm:min_ent}, we get that

\begin{corollary}
$$R_1\subset\left\{0,\log\frac{1+\sqrt5}{2}, h(SFT(101))\right\} \cup\left[\frac{\log(31)}{6},+\infty\right[$$ \end{corollary}

For a subshift $X$ we say that $v$ is a {\emph{minimal prohibited
word}} if $v\notin L_X$ and for each proper factor $u$ of $v$ we
have $u\in L_X$.

\begin{proposition}\label{prop:prohibited_factors}
Let $X$ be a 1-block-gluing subshift over $\{0,1\}$ and let $v$ be
a minimal prohibited word. Then there exists an integer $i\geq 1$
such that, up to renaming letters, $1^i0$ is a prefix of $v$ and
$01^i$ is a suffix of $v$. If in addition $i=1$, then there exists
an integer $j\geq 1$ such that, up to renaming letters, $10^j1$ is
a prefix and a suffix of $v$.
\end{proposition}

The proof of the proposition is split into several lemmas.

\begin{lemma}
Under the conditions of Proposition \ref{prop:prohibited_factors},
we have $v_1=v_n$, where $|v|=n$.
\end{lemma}

\begin{proof}
Indeed, consider a word $w=v_1 \cdots v_{n-1} * v_2 \cdots v_{n}$,
which must be in $L_X$ since $X$ is $1$-block-gluing. If we have
$*=v_1$, then we have an occurrence of $v$ as a suffix, and if we
have $*=v_n$, then we have an occurrence of $v$ as a prefix.
\end{proof}

\begin{lemma}
Under the conditions of Proposition \ref{prop:prohibited_factors},
if $v$ has a prefix and a suffix $1^i$ for $i\geq 1$, then
$v_{i+1}=v_{n-i}$, where $|v|=n$.
\end{lemma}

\begin{proof}
Suppose that $v_{i+1}\neq v_{n-i}$. Up to symmetry we may assume
that $v_{i+1}=1$, $v_{n-i}=0$. Consider a word $w=v_1 \cdots
v_{n-i-1} * v_2 \cdots v_{n}=v_1 \cdots v_{n-i-1} * 1^i \cdots
v_n$, which must be in $L_X$ since $X$ is $1$-block-gluing. If we
have $*=1$, then we have an occurrence of $v$ as a suffix, and if
we have $*=0$, then we have an occurrence of $v$ as a prefix.
\end{proof}

\begin{lemma}
Under the conditions of Proposition \ref{prop:prohibited_factors},
if $v$ has a prefix $10^j$ and a suffix $0^j1$  for $j\geq 1$,
then $v_{j+2}=v_{n-j-1}$, where $|v|=n$.
\end{lemma}

\begin{proof}
Suppose that $v_{j+2}\neq v_{n-j-1}$. Up to symmetry we may assume
that $v_{j+2}=1$, $v_{n-j-1}=0$. Consider a word $w=v_1 \cdots
v_{n-j-2} * v_2 \cdots v_{n}=v_1 \cdots v_{n-j-2} * 0^j \cdots
v_n$, which must be in $L_X$ since $X$ is $1$-block-gluing. If we
have $*=1$, then we have an occurrence of $v$ as a suffix, and if
we have $*=0$, then we have an occurrence of $v$ as a prefix.
\end{proof}

The proof of Proposition \ref{prop:prohibited_factors} follows
from the previous three lemmas.

\bigskip

{\emph{Proof of Theorem \ref{thm:second}.}} We will make use of
the following  claims (the proofs are straightforward based on the ideas we used above):

\bigskip

Claim 1. $h(SFT(111))>h(SFT(101))$.

\bigskip

%Claim 2.  SFT$(10^k1)$, SFT$(1^k)$ are maximal for $k\geq 1$.

%TODO: PROOF HERE OR ELSEWHERE.

%\bigskip

Claim 2. For each 1-block-gluing subshift $Y\neq SFT(11), SFT (101)$ and symmetric ones, with $h(Y)<h(SFT(101))$, one has $\{0,1\}^3\subseteq
L_Y$.

\bigskip

Now, to prove the theorem, suppose that $Z\neq SFT(101), SFT(010)$
is a 1-block-gluing subshift with $h(SFT(11)) < h(Z) <
\frac{\log(31)}{6}$. Clearly, $Z$ must contain all the factors of
length 2 (otherwise either it is not 1-block-gluing or it is
SFT(11) due to Proposition \ref{prop:prohibited_factors}).

It must also contain all the factors of length 3 due to Claim 2.
Now if there are some prohibited factors of length 4, they are in
the set $0^4$, $1^4$, $1001$, $0110$ due to Proposition
\ref{prop:prohibited_factors}. By Proposition \ref{prop:maximal}, only
one of them can be prohibited, and the corresponding SFT is maximal.
The entropies of SFT's can be calculated in a standard way, and $h(SFT(1111)),h( SFT(1001))>\frac{\log{31}}{6}$% (RIGHT?) Yes, h(SFT(1111))=0.656255979236976 h(SFT(1001))=0.624204521655461, log(31)/6=0.57233120074
, so $Z$
must contain all the factors of length 4.

Now consider factors of length 5. By Proposition
\ref{prop:prohibited_factors}, one has that they are from the set
$1^5, 11011, 10101, 10001, 0^5, 00100, 01010, 01110$. Due to
Proposition \ref{prop:lowerbound} for $k=5$, we must prohibit at
least 2 factors of length 5 - and we will now prove that we cannot
prohibit more than one. First note that we cannot prohibit two
factors $u$ and $v$ such that $u$ begins and ends with 0 and $v$
begins and ends with 1. Indeed, consider a word $u^-*^-v$ (i.e.,
$u_1 u_2 u_3 u_4 * v_2 v_3 v_4 v_5$): it has to be in $L_Z$, so at
least one of the words $u$ and $v$ must also be in $L_Z$. So,
without loss of generality, the prohibited words must be from
$1^5, 11011, 10101, 10001$. We cannot prohibit $1^5$ or $10001$,
since they are maximal (by Proposition
\ref{prop:prohibited_factors}); and we cannot prohibit at the same
time $10101$ and $11011$ (consider $101*1011$). So, we can
prohibit at most one factor of length 5, which gives the proof. \qed

\subsection{Minimal entropy for $c=2$ and general observations}

We conjecture that the minimal entropy of binary $c$-block gluing subshifts is given by
$SFT(11, 101,\dots 10^{c-1}1)$ (see Conjecture \ref{conj:main}), but we proved it only for $c=1$ and $c=2$. One of the difficulties in the proofs  of minimality for $c>1$ comes from the increasing number of cases to study.  The other one comes from the fact the subshift with minimal entropy is not unique.

\begin{proposition} The subshifts $X=SFT(11,101,1001)$ and $Y=SFT(000,010,101)$ are 3-block gluing. Moreover,  $h(X)=h(Y)$.\end{proposition}
\begin{proof}
 The property of $3$-block gluing can be straightforwardly checked by Proposition \ref{prop:SFTcbg}. The equiality of the entropies follows from the fact that they are conjugate: More precisely, their  Rauzy graphs $G_3(X)$ and $G_2(Y)$ are isomorphic.
\end{proof}

\begin{proposition} The subshifts $X=SFT(11,101)$ and $Y=SFT(000,101)$ are 2-block gluing. Moreover,  $h(X)=h(Y)$.\end{proposition}

\begin{proof} As in the previous proposition, it is straightforward to check that they are  $2$-block gluing and their Rauzy
graphs $G_4(X)$ and $G_3(Y)$ are isomorphic. \end{proof}

For $c=2$, we proved that the minimal entropy is an isolated point in $R_2$, as it is in the case of $c=1$:
\begin{theorem}
For $c=2$, the minimal positive entropy is given by $SFT(11,101)$, and it is an isolated point in $R_2$.\end{theorem}

%{\color{red}TODO: give an explicit interval without entropies}

%\bigskip

The proof of this theorem is similar to the proof of Theorems \ref{thm:min_ent} and \ref{thm:second}. We omit the details of the thorough case study, but instead we provide a general semi-algorithm to check that $X_c=SFT(11,101, \dots, 10^{c-1}1)$
gives the minimum of the entropy for $c$-block gluing codes and
is separated by a gap from the next one. It also gives some gap in $R_c$, and could easily be adapted for finding the next entropy and proving some related results. However, this is not a proper algorithm since theoretically it might work forever in the case if the minimum of the entropy is
a point of accumulation, otherwise it stops and gives the result (either a proof of minimality, or a subshift with smaller entropy, if the subshift in question does not give a minimal entropy).

\bigskip

First we remark that for each $n$, Proposition \ref{prop:lowerbound} and the entropy $h(X_c)$ give an upper bound $D_n$ on $|L_Y(n)|$ for a subshift $Y$ with $h(Y)\leq h(X_c)$. Indeed, by Proposition \ref{prop:lowerbound}, each $c$-block-gluing subshift satisfies
$$h(Y)\geq  \frac{\log (|L_Y(n)|)}{c+n}$$ for each $n$. So, if $h(Y)\leq h(X_c)$, then  $$\frac{\log (|L_Y(n)|)}{c+n}\leq h(X_c).$$
So, if we define $D_n$ by
$$D_n= \lfloor exp((c+n)h(X_c))\rfloor,$$
 then we must have $L_n(Y)\leq D_n$.

\bigskip

The semi-algorighm works as follows. Start with $n=c+1$ and repeat the following steps:

\begin{enumerate} \item For order $n$, consider the set $S_0(n)$ of all possible SFT's of this order.

\begin{enumerate}\item Choose among them those the set $S_1(n)$ of subshifts satisfying  $L_n(Y)\leq D_n$.

\item For $Y\in S_1(n)$, if there exists a $c$-block gluing subshift $Z$ coinciding with $Y$ on
words of length $n$, then $Y$ is $c$-block gluing by Proposition \ref{prop:Gn}. Clearly, we have $h(Z)\leq h(Y)$. By Proposition \ref{prop:SFTcbg}, we can choose the set
$S_2(n)\subseteq S_1(n)$ of SFT's which are $c$-block gluing.

\item We let $S_2'(n)$ denote the subset of subshifts from $S_2(n)$ which satisfy the sufficient condition on maximality from Remark \ref{remark:maximal}.
%condition  maximal of order $n$, if any (we can check them by Proposition \ref{prop:SFT_max}) {\color{red}skip or modify this step if the proposition is not repaired}.
Let $C_n=\min\{h(Y)|Y\in S_2'(n)\}$. Clearly, if $X_c$ has minimal entropy, then $C_n\geq h(X_c)$, otherwise we found a $c$-block gluing subshift with smaller entropy. So, if $C_n\geq h(X_c)$, then define $S_3$ by $S_3(n)=S_2(n)\setminus S_2'(n)$.\end{enumerate}

\item Now proceed to the next order $n+1$. Define $S_0(n+1)$ as the set of SFT's of order $n+1$ such that their languages of length $n$ coincide with languages of some subshift from $S_3(n)$: $S_0(n+1)=\{Y| \exists Z\in S_3(n): L_n(Y)=L_n(Z)\}$.
 Define its subsets  $S_1(n+1)$, $S_2(n+1)$, $S_3(n+1)$ and $C_{n+1}$ as on the steps 1a,1b,1c. \end{enumerate}

 Possible results of this process are the following:

 \begin{itemize}
     \item  The process stops if at some $n$ the set $S_3(n)$ is empty. In this case we proved that $X_c$ has a minimal entropy, $h(X_c)$ is an isolated point in $R_c$, and we have a lower bound for the gap in $R_c$: the next entropy is at least $\min\{C_{c+1}, \dots, C_n, \frac{\log( D_{c+1}+1)}{c+2}, \dots, \frac{\log(D_n+1)}{n+c}\}$. If the minimum is attained on $C_i$, the corresponding SFT gives the second entropy.

 \item If at some $n$ we receive an element $Y$ from $S_2(n)$ with $h(Y)<h(X_c)$, then $X_c$ was not a subshift with minimal positive entropy among $c$-block gluing subshifts.

 \item The process might as well last forever. That would mean that $X_c$ indeed has minimal entropy, but it is an accumulation point, so we will never know it.
 \end{itemize}

 We suggest that the outcome is the first case for any $c$.
 \bigskip

 \begin{remark}We remark that, given first several values of the entropies of $c$-block gluing subshifts, the semi-algorithm can be easily modified to find the next one (if exists). Another observation is that the step 2c is not actually necessary, but it allows to cut quite a few branches of case study.  \end{remark}
  \bigskip

% {\color{red}TODO: rewrite in a more general form}

% \bigskip

%Note that for manual proof we can use the following:

%- When we build $S_0$, each time some of the words must be present as a unique continuation.
%Also, some must be present because of gluing together certain words
%(in general, this can be checked by the general algorithm for checking if an SFT is c-bg).

%- When we build $S_1$, some of the SFTs can be maximal. If we prove that one of them is,
%then we can calculate its entropy directly, and if it is bigger than $h(X)$,
%then no need to check it further (and  take it into account for the gap).

\section{Conclusions and future work}

The general goal of this research is to characterize the spectrum for the entropies of $c$-block gluing subshifts.
We showed that $R=\cup_{c\in\mathbb{N}} R_c$ is dense, while $R_1$ and $R_2$ are not; a natural conjecture is that this is true for any $c$. However, we proved that for each $c$ the spectrum $R_c$ has infinitely many accumulation points. We also suggest the folowing two conjectures about the minimal entropy in $R_c$, which we proved for $c=1$ and $c=2$:

\begin{conjecture}\label{conj:main}
The minimal positive entropy of binary $c$-block gluing subshifts is given by
$$X_c=SFT(11, 101,\dots 10^{c-1}1).$$
\end{conjecture}

We remark that for $c>1$ the subshift $X_c$ is not the unique $c$-block gluing subshift with this entropy.

\begin{conjecture}\label{conj:main1}
The entropy of
$X_c$ is an isolated point in $R_c$.
\end{conjecture}

%We verified the conjectures for $c=1$ and $c=2$.
%We showed that $R=\cup_{c\in\mathbb{N}} R_c$ is dense, while $R_1$ and $R_2$ are not; a natural conjecture is that this is true for any $c$. However, we proved that for each $c$ the spectrum $R_c$ has infinitely many accumulation points.

One of the open questions on the structure of $R_c$ is the following:

\begin{question} Given $c$, does $R_c$ have an interval of density?\end{question}

%%
%% Bibliography
%%

\bibliography{main_journal}

\end{document}